\documentclass[12pt,reqno]{amsart}

\usepackage{amsmath,amssymb,latexsym,amscd,amsthm} 
\usepackage{graphicx}
\usepackage{booktabs}
\usepackage[square, numbers]{natbib}   
\usepackage{fourier}
\usepackage{comment}

\newtheorem{theorem}{Theorem}[section]  
\newtheorem{corollary}[theorem]{Corollary}
\newtheorem{proposition}[theorem]{Proposition}

\theoremstyle{definition}

\theoremstyle{remark}
\newtheorem{remark}[theorem]{Remark}

\title[A further generalization of the colourful Carath\'eodory theorem]{A further generalization of the colourful\\ Carath\'eodory theorem}
\date{March 5th, 2014}

\author{Fr\'ed\'eric Meunier}
\address{Universit\'e Paris Est, CERMICS, 6-8 avenue Blaise Pascal, Cit\'e Descartes, 77455 Marne-la-Vall\'ee, Cedex 2, France}
\email{frederic.meunier@cermics.enpc.fr}
\author{Antoine Deza}
\address{Advanced Optimization Laboratory, Department of Computing and Software, 
McMaster University, Hamilton, Ontario, Canada, 
\emph{and} Equipe Combinatoire et Optimisation, Universit\'e Pierre et Marie Curie, Paris, France}
\email{deza@mcmaster.ca}



\def\bara{B\'ar\'any}

\def\cara{Carath\'eodory}

\def\R{\mathbb{R}}
\def\vecr{\mathbf{r}}
\def\S{\mathbf{S}}

\def\C{\mathcal{C}}

\def\F{F}

\def\Sph{\mathbb{S}}

\def\conv{\operatorname{conv}}

\def\zero{{\bf 0}}
\def\aff{\operatorname{aff}}

\begin{document}

\begin{abstract}
Given $d+1$ sets, or colours, $\S_1, \S_2,\ldots,\S_{d+1}$ of points in $\R^d$, a {\em colourful} set is a set $S\subseteq\bigcup_i\S_i$ such that $|S\cap\S_i|\leq 1$ for $i=1,\ldots,d+1$. The convex hull of a colourful set $S$ is called a {\em colourful simplex}. {\bara}'s colourful {\cara} theorem asserts that if the origin {\zero} is contained in the convex hull of $\S_i$ for $i=1,\ldots,d+1$, then there exists a colourful simplex containing {\zero}. The sufficient condition for the existence of a colourful simplex containing {\zero} was generalized to {\zero} being contained in the convex hull of  $\S_i\cup\S_j$ for $1\leq i< j \leq d+1$ by Arocha et al. and by Holmsen et al. We further generalize the sufficient condition and obtain new colourful {\cara} theorems. We also give an algorithm to find a colourful simplex containing {\zero} under the generalized condition. In the plane an alternative, and more general, proof using graphs is given. In addition, we observe that any condition implying the existence of a colourful simplex containing {\zero} actually implies the existence of $\min_i|\S_i|$ such simplices.
\end{abstract}
\maketitle

\section{Colourful {\cara} theorems}
Given $d+1$ sets, or colours, $\S_1, \S_2,\ldots,\S_{d+1}$ of points in $\R^d$, we call a set of points drawn from the $\S_i$'s {\em colourful} if it contains at most one point from each $\S_i$. A {\em colourful simplex} is the convex hull of a colourful set $S$, and a colourful set of $d$ points which misses $\S_i$ is called an $\widehat{i}$-{\it transversal}. The colourful {\cara} Theorem~\ref{thmB} by {\bara} provides a sufficient condition for the existence of a colourful simplex containing the origin {\zero}.
\begin{theorem}[\cite{Bar82}] \label{thmB}
Let $\S_1,\S_2,\ldots,\S_{d+1}$ be finite sets of points in $\R^d$ such that $\zero\in\conv(\S_i)$ for $i=1\ldots d+1$. Then there exists a set $S\subseteq\bigcup_i\S_i$ such that $|S\cap\S_i|=1$ for $i=1,\ldots,d+1$ and $\zero\in\conv(S)$.
\end{theorem}
\noindent
Theorem~\ref{thmB} was generalized by Arocha et al.~\cite{AB+09} and by Holmsen et al.~\cite{HPT08} providing a more general sufficient condition for the existence of a colourful simplex containing the origin {\zero}.
\begin{theorem}[\cite{AB+09,HPT08}] \label{thmB+}
Let $\S_1,\S_2,\ldots,\S_{d+1}$ be finite sets of points in $\R^d$ such that $\zero\in\conv(\S_i\cup\S_j)$ for $1\leq i<j\leq d+1$. Then there exists a set $S\subseteq\bigcup_i\S_i$ such that $|S\cap\S_i|=1$ for $i=1,\ldots,d+1$ and $\zero\in\conv(S)$.
\end{theorem}
\noindent
We further generalize the sufficient condition for the existence of a colourful simplex containing the origin. Moreover, the proof, given in Section~\ref{Sproofsub}, provides an alternative and geometric proof for Theorem~\ref{thmB+}. Let $\overrightarrow{x_k\zero}$ denote the ray originating from $x_k$ towards {\zero}.
\begin{theorem}\label{thm1}
Let $\S_1,\S_2,\ldots,\S_{d+1}$ be finite sets of points in $\R^d$. Assume that, for each $1\leq i< j \leq d+1$, there exists $k\notin\{i,j\}$ such that, for all $x_k\in\S_k$, the convex hull of $\:\S_i\cup\S_j$ intersects the ray $\overrightarrow{x_k\zero}$ in a point distinct from $x_k$. Then there exists a set $S\subseteq\bigcup_i\S_i$ such that $|S\cap\S_i|=1$ for $i=1,\ldots,d+1$ and $\zero\in\conv(S)$.
\end{theorem}
%
\begin{figure}[htb] 
\includegraphics[width=7cm]{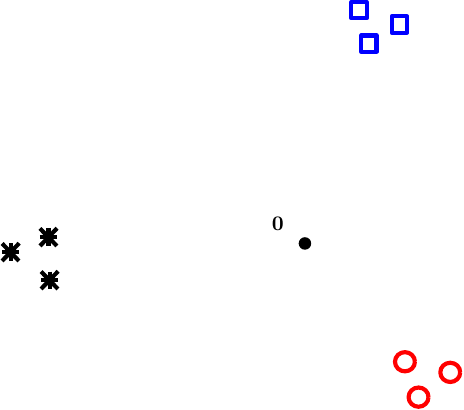}
\caption{A set in dimension $2$ satisfying the condition of Theorem~\ref{thm1} but not the one of Theorem~\ref{thmB+}.}
\label{fig:2d}
\end{figure}
\noindent
Under the general position assumption, Theorem~\ref{thm1} can be derived from the slightly stronger Theorem~\ref{thm2}  where $H^+(T_{i})$ denotes, for any $\widehat{i}$-transversal $T_{i}$, the open half-space defined by $\aff(T_{i})$ and containing {\zero}. 
\begin{theorem}\label{thm2}
Let $\S_1,\S_2,\ldots,\S_{d+1}$ be finite sets of points in $\R^d$ such that the points in $\bigcup_i\S_i\cup\{\zero\}$ are distinct and in general position. Assume that, for any $i\neq j$,  $(\S_i\cup\S_j)\cap H^+(T_j)\neq\emptyset$ for any $\widehat{j}$-transversal $T_j$.  Then there exists a set $S\subseteq\bigcup_i\S_i$ such that $|S\cap\S_i|=1$ for $i=1,\ldots,d+1$ and $\zero\in\conv(S)$.
\end{theorem}
\noindent
Note that, as the conditions of Theorems~\ref{thmB} and~\ref{thmB+}, but unlike the one of Theorem~\ref{thm2}, the condition of Theorem~\ref{thm1} is computationally easy to check. Indeed, testing whether a ray intersects the convex hull of a finite number of  points amounts to solve a linear optimization feasibility problem which is polynomial-time solvable.\\

\noindent
In the plane and assuming general position, Theorem~\ref{thm1} can be generalized to Theorem~\ref{thm2d}. The proofs of Theorems~\ref{thm1},~\ref{thm2}, and~\ref{thm2d} are given in Section~\ref{Sproof}.
\begin{theorem}\label{thm2d}
Let $\S_1,\S_2,\S_3$ be finite sets of points in $\R^2$ such that the points in $\S_1\cup\S_2\cup\S_3\cup\{\zero\}$ are distinct and in general position.
Assume that, for pairwise distinct $i,j,k\in\{1,2,3\}$, the convex hull of $\:\S_i\cup\S_j$ intersects the line $\aff(x_k,\zero)$ for all $x_k\in\S_k$.
Then there exists a set $S\subseteq\S_1\cup\S_2\cup\S_3$ such that $|S\cap\S_i|=1$ for $i=1,2,3$ and $\zero\in\conv(S)$.
\end{theorem}
\noindent
Figures~\ref{fig:2d} and~\ref{fig:3d} illustrate sets satisfying the condition of Theorem~\ref{thm1} but not the ones of Theorems~\ref{thmB} and~\ref{thmB+}. 
Let $S_d^{\triangle}$ denote the $d$-dimensional configuration where the points in $\S_i$ are clustered around the $i^{th}$ vertex of a simplex containing {\zero},
see Figure~\ref{fig:2d} for an illustration of $S_2^{\triangle}$. 
While all the $(d+1)^{d+1}$ colourful simplices of this configuration  contain {\zero},  $S_{d\geq 3}^{\triangle}$ does not satisfy the conditions of Theorems~\ref{thmB},~\ref{thmB+}, or~\ref{thm1}, but  satisfies the one of Theorem~\ref{thm2}. 
While the set given in  Figure~\ref{fig:3d_gen} satisfies the condition of Theorem~\ref{thm2}, it  does not satisfy the condition of  Theorem~\ref{thm1}  for $i=$~\includegraphics[width=0.3cm]{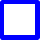} and $j=$~\includegraphics[width=0.3cm]{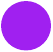}.
Figure~\ref{fig:2d_gen} illustrates a set satisfying the condition of Theorem~\ref{thm2d} but not the one of Theorem~\ref{thm2}.\\

\begin{figure}[htb]  
\includegraphics[width=8cm]{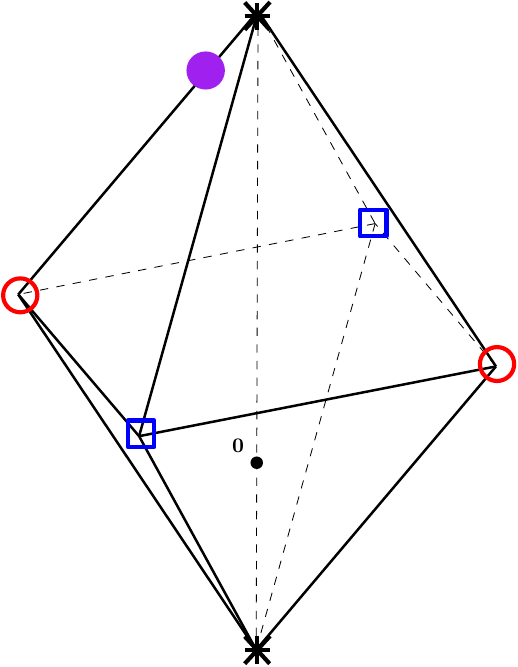}
\caption{A degenerate set in dimension $3$ satisfying the condition of Theorem~\ref{thm1} but not the one of Theorem~\ref{thmB+}.}
\label{fig:3d}
\end{figure}

\noindent
One can check that Theorem~\ref{thm2} is still valid if the general position assumption is replaced by: {\em there is at least one transversal $T$ such that $\zero\notin\aff(T)$ and such that the points of $T$ are affinely independent}. However, we are not aware of an obvious way to handle, via  Theorem~\ref{thm2}, configurations where all points and the origin lie in the same hyperplane. Note that Theorem~\ref{thm1} can be applied to such degenerate configurations. See  Section~\ref{subsec:thm1} for a proof of Theorem~\ref{thm1} and a configuration which illustrates the gap between Theorem~\ref{thm1} and its general position version, and justifies the specific treatment for the degenerate cases.

\begin{figure}[htb]  
\includegraphics[width=8cm]{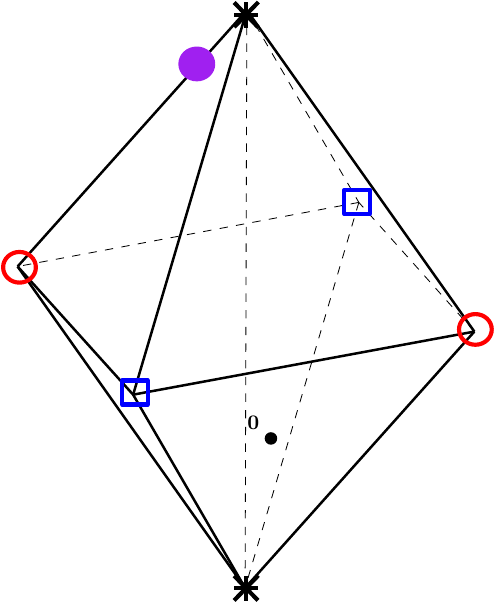}
\caption{A degenerate set in dimension $3$ satisfying, up to a slight perturbation, the condition of Theorem~\ref{thm2} but not the one of Theorem~\ref{thm1}.}
\label{fig:3d_gen}
\end{figure}

\begin{figure}[htb]  
\includegraphics[width=7cm]{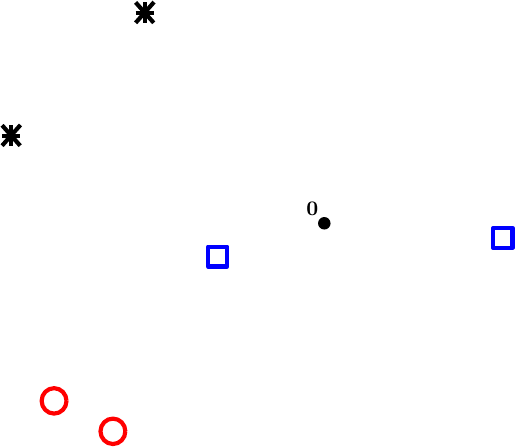}
\caption{A  set in dimension $2$ satisfying the condition of Theorem~\ref{thm2d} but not the one of Theorem~\ref{thm2}.}
\label{fig:2d_gen}
\end{figure}

\clearpage
\newpage
\section{Proofs}\label{Sproof}
\subsection{Proof of Theorem~\ref{thm2}}\label{Sproofsub}
We recall that a $k$-simplex $\sigma$ is the convex hull of $(k+1)$ affinely independent points. 
An {\em abstract simplicial complex} is a family $\mathcal{F}$ of subsets of a finite ground set such that whenever $\F\in\mathcal{F}$ and $G\subseteq \F$, then $G\in\mathcal{F}$. These subsets are called {\em abstract} simplices. The {\em dimension} of an abstract simplex is its cardinality minus one. The {\em dimension} of a simplicial complex is the dimension of largest simplices. A {\em pure} abstract simplicial complex is a simplicial complex whose maximal simplices have all the same dimension. A {\em combinatorial $d$-pseudomanifold} $\mathcal{M}$ is a pure abstract $d$-dimensional simplicial complex such that any abstract $(d-1)$-simplex is contained in exactly $2$ abstract $d$-simplices.\\

\noindent
Consider a ray $\vecr$ originating from $\zero$ and intersecting at least one colourful $(d-1)$-simplex. 
Under the general position assumption for points in $\bigcup_i\S_i\cup\{\zero\}$, one can choose $\vecr$ such that it intersects the interior of
the colourful $(d-1)$-simplex, and that no two colourful simplices have the same intersection with $\vecr$. 
Let $\sigma$ be the {\em first} colourful $(d-1)$-simplex intersected by $\vecr$. Note that, given $\vecr$, $\sigma$ is uniquely defined.
Without loss of generality, we can assume that the vertices of $\sigma$ form the $\widehat{d+1}$-{\it transversal} $\{v_1,\ldots,v_d\}$.\\

\noindent
Setting $j=d+1$, and $T_{d+1}=\{v_1,\ldots,v_d\}$ in Theorem~\ref{thm2} gives  $(\S_i\cup\S_{d+1})\cap H^+(T_{d+1})\neq\emptyset$. 
In other words, there is, for each $i$, a point either in $\S_{d+1}\cap H^+(T_{d+1})$ or in $(\S_i\setminus\{v_i\})\cap H^+(T_{d+1})$. 
Assume first that for one $i$ the corresponding point belongs to $\S_{d+1}$, and name it $v'_{d+1}$.
Then $\vecr$ intersects the boundary of $\conv(v_1,\ldots,v_d,v'_{d+1})$ in only one point
as otherwise $\vecr$ would intersect another colourful $(d-1)$-simplex before intersecting $\sigma$.
Indeed, $\vecr$ leaves $H^+(T_{d+1})$ after intersecting $\sigma$. 
Thus, $\vecr$ intersects $\conv(v_1,\ldots,v_d,v'_{d+1})$ in exactly one point; that is, $\zero\in\conv(v_1,\ldots,v_d,v'_{d+1})$.
Therefore, we can assume that for each $i$ there is a point $v'_i\neq v_i$ in $\S_i\cap H^+(T_{d+1})$, and consider the 
$\widehat{d+1}$-transversal $T'=\{v'_1,\ldots,v'_d\}$ and the associated colourful $(d-1)$-simplex $\sigma'=\conv(v'_1,\ldots,v'_d)$.
Let $\mathcal{M}$ be the abstract simplicial complex defined by
	$$\mathcal{M}=\{F\cup F':\,F\subseteq V(\sigma),\,F'\subseteq V(\sigma')\mbox{ and }c(F)\cap c(F')=\emptyset\}$$
 where $V(\sigma)$ denotes the vertex set of $\sigma$ and $c(x)=i$ for $x\in\S_i$. 
The simplicial complex $\mathcal{M}$ is a combinatorial $(d-1)$-pseudomanifold. Note that $V(\sigma)$ and $V(\sigma')$ are abstract simplices of $\mathcal{M}$.
Let $M$ be the collection of the convex hulls of the abstract simplices of $\mathcal{M}$.  Note that the vertices of all maximal simplices of $M$ form $\widehat{d+1}$-transversals and that $\mathcal{M}$
is not necessarily a simplicial complex in the geometric meaning as some pairs of geometric $(d-1)$-simplices might have intersecting interiors.\\

\noindent
We recall that for any generic ray originating from $\zero$, the parity of the number of times its intersects $M$ is the same. 
We remark that this number can not be even as, otherwise, we would have a colourful $(d-1)$-simplex closer to $\zero$ than $\sigma$ on $\vecr$ since,
$M$ being contained in the closure of $H^+(T_{d+1})$, when $\vecr$ intersects $\sigma$, it is the last intersection. 
Thus, the number of times $\vecr$ intersects $M$ is odd, and actually equal to $1$. 
Take now any point $v\in\S_{d+1}$ and consider the ray originating from $\zero$ towards the direction opposite to $v$. 
This ray intersects $M$ in a colourful $(d-1)$-simplex $\tau$; that is, $\zero\in\conv(\tau\cup\{v\})$.
\qed\\

\noindent
One can check that the proof of Theorem~\ref{thm2} still works if there is at least one transversal $T$ such that $\zero\notin\aff(T)$ and such that the points of $T$ are affinely independent.  Indeed, in that case, we can always choose a ray $\vecr$ such that, for any pair $(T,T')$ of transversals, $\vecr\cap\aff(T) = \vecr \cap \aff(T')$ if and only if $\aff(T)=\aff(T')$. 

\begin{remark}
The topological argument that the parity of the number of times a ray originating from $\zero$ intersects $M$ depends only on the respective positions of $\zero$ and $M$ can be replaced by Proposition~\ref{prop2} as used in the description of the algorithm given in Section~\ref{Salgo}. In other words, we get a geometric proof of Theorem~\ref{thm2}.
\end{remark}

\noindent
Assuming $\bigcup_i\S_i$ lies on the sphere $\Sph^{d-1}$, the $\widehat{i}$-transversals generate full dimensional colourful cones pointed at {\zero}. 
We say that a transversal {\it covers} a point if the point is contained in the associated cone.
Colourful simplices containing $\zero$ are generated whenever the antipode of a point of colour $i$ is covered by an $\widehat{i}$-transversal.
In particular, one can consider combinatorial {\it octahedra} generated by pairs
of disjoint $\widehat{i}$-transversals, and rely on the fact that every octahedron $\Omega$ either
covers all of $\Sph^{d-1}$ with colourful cones, or every point $x \in \Sph^{d-1}$ that is covered by colourful cones 
from $\Omega$ is covered by at least two distinct such cones, see for example the {\em Octahedron Lemma} of~\cite{BM06}. 
One of the key argument in the proof of Theorem~\ref{thm2} can be reformulated as: either the pair of $\widehat{d+1}$-transversals $(T,T')$ forms a octahedron covering $\Sph^{d-1}$, or {\zero} belongs to a colourful simplex having $\conv(T)$ as a facet.

\subsection{Proof of Theorem~\ref{thm1}}\label{subsec:thm1}
Consider a configuration satisfying the conditions of Theorem~\ref{thm1} and with $\bigcup_i\S_i\cup\{\zero\}$ distinct and in general position. Consider $i\neq j$ and a ${\hat j}$-transversal $T_j$, then there is $x_k\in\S_k\cap T_j$ such that the $\overrightarrow{x_k\zero}$ intersect the convex hull of $\:\S_i\cup\S_j$ in a point in $H^+(T_j)$, and therefore at least one point of $\:\S_i\cup\S_j$ belongs to $H^+(T_j)$.\\

\noindent 
Let consider degenerate  configurations and let $a$ denote the maximum cardinality of an affinely independent colourful set whose affine hull does not contain {\zero}.\\

\noindent
 If $a=d$, there is at least one transversal $T$ such that $\zero\notin\aff(T)$ and such that the points of $T$ are affinely independent. Therefore we can use the stronger version of Theorem~\ref{thm2} relying on the existence of such a  transversal $T$.\\

\noindent
Assume that $a<d$. We can choose a ray $\vecr$ such that the non-empty intersections with $\aff(A)$ for all colourful sets $A$ of cardinality $a$ are distinct. Let $A^0$ be an affinely independent colourful set of cardinality $a$ such that $\aff(A^0)$ is the first intersected by $\vecr$. Without loss of generality, let $A^0=\{v_1,\ldots,v_a\}$ with $v_s\in\S_s$. Note that $\S_{a+1}\cup\ldots\cup\S_{d+1}\subset\aff(A^0\cup\{\zero\})$ as otherwise $\zero\notin\aff(A^0\cup\{v_j\})$ for $v_j\in\S_j$ with $j>a$ which contradicts the maximality of $a$.  If there is a colourful simplex containing {\zero}, we are done. Therefore, we can assume that, in $\aff(A^0\cup\{\zero\})$, we have an open half-space defined by $\aff(A^0)$ containing {\zero} but not $\S_{a+1}\cup\ldots\cup\S_{d+1}$, and will derive a contradiction.\\

\noindent
Let $B_0=\{a+1,\ldots,d+1\}$. We remark that, for all $i,j\in B_0$ with $i\neq j$, the $k$, such that $\conv(\S_i\cup\S_j)$ intersects $\overrightarrow{x_k\zero}$ in a point distinct from $x_k$, satisfies $k\in B_0$ since $\S_i\cup\S_j$ are separated from {\zero} by $\aff(A^0)$ in $\aff(A^0\cup\{\zero\})$; and therefore we have $|B_0|\geq 3$. We can define the following set map:
$$\mathcal{F}(B)=
\left\{
\begin{array}{ll}
\{k: \exists (i,j)\in B\times B,\,i\neq j, \forall x_k\in\S_k,\,\conv(\S_i\cup\S_j)\cap\overrightarrow{x_k\zero}\setminus\{x_k\}\neq\emptyset\} & \mbox{if $|B|\geq 2$} \\
 \emptyset & \mbox{otherwise.}
\end{array}
\right.$$ 
We have $\mathcal{F}(B)\subseteq\mathcal{F}(B')$ if $B\subseteq B'$. Let $B_{\ell}=\mathcal{F}(B_{\ell-1})$ for $\ell=1,2,\dots$
As remarked above $B_1\subseteq B_0$ and, by induction, $B_{\ell}\subseteq B_{\ell-1}$ for $\ell\geq 1$. 
Thus, the sequence $(B_{\ell})$ converges towards a set $B^*$ satisfying $\mathcal{F}(B^*)=B^*$. 
Finally, note that, by induction, $|B_{\ell}|\geq 3$: The base case holds as $|B_0|\geq 3$, and a pair $i,j\in B_{\ell}$ with $i\neq j$ yields a $k\in B_{\ell+1}$, then $i,k$ yields an additional $k'$ in $B_{\ell+1}$, which in turn, with $k$, yields a third element in $B_{\ell+1}$; and thus $|B^*|\geq 3$.\\

\noindent
For any $v\in\bigcup_{k\in B^*}\S_k$, the ray $\overrightarrow{v\zero}$ intersects the convex hull of $\bigcup_{k\in B^*}\S_k$ in a point distinct from $v$ since $\mathcal{F}(B^*)=B^*$. It contradicts the fact that $\aff(A^0)$ separates {\zero} from $\S_{a+1}\cup\ldots\cup\S_{d+1}$ in $\aff(A^0\cup\{\zero\})$ by the following argument. There exists at least one facet of $\conv(\bigcup_{k\in B^*}\S_k)$ whose supporting hyperplane separates {\zero} from $\conv(\bigcup_{k\in B^*}\S_k)$ and, for a vertex $v$ of this facet, we have $\conv(\bigcup_{k\in B^*}\S_k)\cap\overrightarrow{v\zero}=\{v\}$, which is impossible.\qed\\

\noindent
The gap between Theorem~\ref{thm1} and its general position version is illustrated by the following example in  $\mathbb{R}^{3}$ where $\bigcup_i\S_i\cup\{\zero\}$ lie in the same plane.
{\em Let $\S_1,\S_2,\S_3,\S_4$ be finite sets of points in $\R^2$.  Assume that, for each $1\leq i< j \leq d+1$, there exists $k\notin\{i,j\}$ such that, for all $x_k\in\S_k$, 
the convex hull of $\:\S_i\cup\S_j$ intersects the ray $\overrightarrow{x_k\zero}$ in a point distinct from $x_k$. Then there exists a set $S\subseteq\bigcup_i\S_i$ such that $|S\cap\S_i|=1$ for $i=1,\ldots,d+1$ and $\zero\in\conv(S)$.} This property cannot be obtained  by simply applying Theorem~\ref{thm1} with $d=2$ since its conditions might not be satisfied by $\S_1,\S_2,\S_3$. Indeed, $k$ may be equal to $4$ for some $i\neq j$. This property can neither be obtained by a compactness argument since it would require to find sequences $(\S_i^j)_{j=1,\ldots,\infty}$ of generic point sets converging to $\S_i$ while satisfying the condition of Theorem~\ref{thm1}. The case when each $\S_i$ is reduced to one point $s_i$ shows that such a sequence may fail to exist as the condition implies that $s_1^j$, $s_2^j$, $s_3^j$, $s_4^j$ and $\zero$  lie in a common plane. This might explain why we could not avoid a Tarsky-type fixed point argument.
\subsection{Proof of Theorem~\ref{thm2d}
}\label{Sproof2}
We present a proof of Theorem~\ref{thm2d} for the planar case providing an alternative and possibly more combinatorial proof of Theorem~\ref{thm2} in the plane.
Consider the graph $G=(V,E)$ with $V=\S_1\cup\S_2\cup\S_3$ and where a pair of nodes are adjacent if and only if they have different colours.
We get a directed graph $D=(V,A)$ by orienting the edges of $G$ such that {\zero} is always on the right side of any arc, i.e. on the right side of the line extending it, with the induced orientation. 
Since $\conv(\S_i\cup\S_j)\cap\aff(x_k,\zero)\neq\emptyset$ with $i,j,k$ pairwise distinct and $x_k\in\S_k$, we have $\deg^+(v)\geq 1$ and $\deg^-(v)\geq 1$ for all $v\in V$. 
It implies that there exists at least one circuit in $D$, and we consider the shortest circuit $C$. 
We first show that the length of $C$ is at most $4$ since any circuit of length $5$ or more has necessarily a chord. Indeed, take a vertex $v$, there is a vertex $u$ on the circuit at distance $2$ or $3$ having a colour distinct from the colour of $v$, and thus the arc $(u,v)$ or $(v,u)$ exists in $D$. 
Therefore, the length of $C$ must be $3$ or $4$. If the length is $3$, we are done as the $3$ vertices of $C$ form a colourful triangle containing $\zero$. 
If the length is $4$, the circuit $C$ is $2$-coloured as otherwise we could again find a chord. 
Consider such a $2$-coloured circuit $C$ of length $4$ and take any generic ray originating from $\zero$. 
We recall that given an oriented closed curve $\C$ in the plane, 
with $k_+$, respectively $k_-$, denoting the number of times a generic ray intersects $\C$ while entering by the right, respectively left, side,  
the quantity $k_+-k_-$ does not depend on the ray.
Considering the realization of $C$ as a curve $\C$, we have $k_-=0$ by definition of the orientation of the arcs. 
Since we can choose a ray intersecting  $C$ at least once,  $k_+$ remains constant and non-zero.
Take now a vertex $w$ of the missing colour, and take the ray originating from $\zero$ in the opposite direction.
This ray intersects an arc of $C$ since $k_+\neq 0$, and the endpoints of the arc together with $w$ form a colourful triangle containing $\zero$.\qed

\begin{remark}
The fact that a directed graph missing a source or a sink has always a circuit is a key argument, and it is not clear to us how the planar proof could be extended or adapted to dimension $3$ or more.
\end{remark}

\section{Related results and an algorithm}
\subsection{Given one, find another one} B\'ar\'any and Onn~\cite{BO97} raised the following algorithmic question:
Given sets $\S_i$ containing $\zero$ in their convex hulls, finding a colourful simplex containing $\zero$ in its convex hull. This question,  called {\em colourful feasibility problem},  belongs to the {\em Total Function Nondeterministic Polynomial} (TFNP) class, i.e. problems whose decision version has always a {\em yes} answer.  The geometric algorithms introduced by {\bara}~\cite{Bar82} and {\bara} and Onn~\cite{BO97}  and other methods to tackle the colourful feasibility problem, such as multi-update modifications, are studied and benchmarked in~\cite{DHST08}. The complexity  of this challenging problem, i.e.  whether it is polynomial-time solvable or not, is still an open question.  However, there are strong indications  that no TFNP-complete problem exists, see~\cite{P94}. The following Proposition~\ref{prop2}, which is similar in flavour to Theorem~\ref{thmB}, may indicate  an inherent hardness result for  this relative of the colourful feasibility  problem. Indeed, the algorithmic problem associated to Proposition~\ref{prop2} belongs to the {\em Polynomial Parity Argument}  (PPA) class defined by Papadimitriou~\cite{P94} for which complete problems are known to exist. In addition, the proof of Proposition~\ref{prop2} is a key ingredient of the algorithm finding a colourful simplex under the condition of Theorem~\ref{thm2}.
\begin{proposition}\label{prop2}
Given $d+1$ sets, or colours, $\S^*_1, \S^*_2,\ldots,\S^*_{d+1}$ of points in $\R^d$ with $|\S^*_i|=2$ for $i=1,\dots,d+1$, if there is a colourful simplex containing $\:\zero$, then
there is another colourful simplex containing $\:\zero$.
\end{proposition}
\begin{proof}
Without loss of generality we assume that the points in $\bigcup_i\S^*_i\cup\{\zero\}$ are distinct and in general position.
Consider the graph $G$ whose nodes consist  of some subsets $\bigcup_i\S^*_i$  partitioned into three types:
$(i)$   $N_1$ made of subsets $\nu_1$ of cardinality $d+2$ with $\zero\in\conv(\nu_1)$, $|\nu_1\cap\S^*_i|=1$ for $i=1,\ldots,d$, and $|\nu_1\cap\S^*_{d+1}|=2$;
$(ii)$  $N_2$ made of subsets $\nu_2$ of cardinality $d+1$ with $\zero\in\conv(\nu_2)$, $|\nu_2\cap\S^*_i|= 1$ for $i=1,\ldots,d$ except for exactly one $i$, and $|\nu_2\cap\S^*_{d+1}|=2$; and
$(iii)$ $N_3$ made of subsets $\nu_3$ of cardinality $d+1$ with $\zero\in\conv(\nu_3)$ and $|\nu_3\cap\S^*_i|=1$ for $i=1,\ldots,d+1$. The adjacency between the nodes of $G$ is defined as follows.
There is no edge between nodes of type $\nu_2$ and $\nu_3$. The nodes $\nu_1$ and $\nu_2$, respectively $\nu_1$ and $\nu_3$, are adjacent if and only if $\nu_2\subseteq \nu_1$, respectively $\nu_3\subseteq \nu_1$.\\

\noindent
We show that $G$ is a collection of node-disjoint paths and cycles by checking the degree of $N_1$, $N_2$, and $N_3$ nodes. First consider a $N_1$ node $\nu_1$. We recall that, under the general position assumption, there are exactly two $d+1$-subsets $\chi$ and $\chi'$ of $\nu_1$ containing $\zero$ in their convex hull. This fact can be expressed as,
using the simplex method terminology, there is a unique leaving variable in a pivot step of the simplex method assuming non-degeneracy.
Both $\chi$ and $\chi'$ intersect $\S^*_i$ for $i=1,\ldots,d$ in at least one point except maybe for one $i$. Thus, $\chi$ and $\chi'$ are $N_2$ or $N_3$ nodes, hence the degree of a $N_1$ node is $2$.
Consider now a $N_2$ node $\nu_2$, there is a $i_0\neq d+1$ such that $|\nu_2\cap\S^*_{i_0}|=0$. The node $\nu_2$ is contained in exactly two $N_1$ nodes, each of them obtained by adding one of the points in $\S^*_{i_0}$. Hence the degree of a $N_2$ node is $2$. 
Finally, consider a $N_3$ node, it is contained in exactly one $N_1$ node obtained by adding the missing point of $\S^*_{d+1}$.
Hence, the degree of a $N_3$ node is $1$. The graph $G$ is thus a collection of node disjoint paths and cycles.\\

\noindent 
Therefore, the existence of a colourful simplex containing $\zero$ provides a $N_3$ node, and following the path in $G$ until reaching the other endpoint provides  another node of degree $1$, i.e. a $N_3$ node corresponding to a distinct colourful simplex containing the origin $\zero$.
\end{proof}
\noindent
Proposition~\ref{prop2} raises the following problem, which we call {\em Second covering colourful simplex}: Given $d+1$ sets, or colours, $\S_1, \S_2,\ldots,\S_{d+1}$ of points in $\R^d$ with $|\S_i|\geq 2$ for $i=1,\dots,d+1$, and a colourful set $S\subseteq\bigcup_i\S_i$ containing $\zero$ in its convex hull, find another such set.
The key property used in the proof of Proposition~\ref{prop2} is the fact that the existence of one odd degree node in a graph implies the existence of another one. In other words, the proof of Proposition~\ref{prop2} shows that {\em Second covering colourful simplex} belongs to the PPA class, which forms precisely the problems in TFNP for which the existence is proven through this parity argument. Other examples of PPA problems include {\em Brouwer}, {\em Borsuk-Ulam}, {\em Second Hamiltonian circuit}, {\em Nash}, or {\em Room partitioning}~\cite{ES10}.
The PPA class has a nonempty subclass of PPA-complete problems for which the existence of a polynomial algorithm would imply the existence of a polynomial algorithm for any problem in PPA, see Grigni~\cite{G01}. We do not know whether {\em Second covering colourful simplex} is PPA-complete, but it is certainly a challenging question related to the complexity of colourful feasibility problem.\\

\noindent
 Note that Proposition~\ref{prop2} can also be proven by a degree argument on the map embedding the join of the $\S^*_i$ in $\R^d$, or using the  Octahedron Lemma~\cite{BM06}.  
 
\subsection{Minimum number of colourful simplices containing {\zero}} As a corollary of Proposition~\ref{prop2}, any condition implying the existence of a colourful simplex containing {\zero} actually implies the existence of $\min_i|\S_i|$ such simplices.
\begin{corollary}\label{atleast}
Given $d+1$ sets, or colours, $\S_1, \S_2,\ldots,\S_{d+1}$ of points in $\R^d$, if there is a colourful simplex containing $\:\zero$, then
there are at least $\min_i|\S_i|$ colourful simplices containing the origin $\:\zero$.
\end{corollary}
\begin{proof}
Let $I=\min_i|\S_i|$ and $\S_i=\{v_i^1,v_i^2,\dots\}$, and assume without loss of generality that the given colourful simplex containing {\zero} in its convex hull is $\conv(v_1^1,v_2^1,\ldots,v_{d+1}^1)$. 
Applying Proposition~\ref{prop2} $(I-1)$ times with $\S^*_i=\{v_i^1,v_i^k\}$ we obtain an additional distinct colourful simplex containing $\zero$ for each $k=2,\dots,I$.
\end{proof}
\noindent

The minimum number $\mu(d)$ of colourful simplices containing {\zero} for sets satisfying the condition of Theorem~\ref{thmB},  the general position assumption, and $|\S_i|\geq d+1$ for all $i$ was investigated in~\cite{BM06,DHST06,DSX11,ST06}. While it is conjectured that $\mu(d)= d^2+1$ for all $d \ge 1$, the best current upper and lower bounds are $d^2+1\ge\mu(d) \ge \left\lceil \frac{(d+1)^2}{2} \right\rceil$. In addition, we have $\mu(3)=10$ and $\mu(d)$ even for odd $d$. For sets satisfying  $|\S_i|\geq d+1$ for all $i$, one can consider the analogous quantities $\mu^{\diamond}(d)$, respectively $\mu^{\circ}(d)$, defined as the
minimum number of colourful simplices containing {\zero} for sets satisfying the condition
of Theorem~\ref{thmB+} and the general position assumption, respectively Theorem~\ref{thm2}. 
Since $\mu^{\circ}(d)\leq\mu^{\diamond}(d)$ and, as noted  in~\cite{CDSX11}, $\mu^{\diamond}(d)=d+1$, Theorem~\ref{thm2} and Corollary~\ref{atleast} imply that $\mu^{\circ}(d)=d+1$ for $d\geq 2$.
\subsection{Algorithm to find a colourful simplex}\label{Salgo}
We present an algorithm based on the proof of Proposition~\ref{prop2} finding a colourful simplex containing {\zero} for sets satisfying the conditions
of Theorem~\ref{thm2}, and, therefore, for sets satisfying the condition of Theorem~\ref{thm1} and the general position assumption. Note that the algorithm also finds a colourful simplex under the condition of Theorem~\ref{thmB+}. 
\subsection*{Algorithm:} Take any colourful $(d-1)$-simplex $\sigma$ whose vertices form, without loss of generality, a $\widehat{d+1}$-transversal $T=\{v_1,\ldots,v_d\}$,
and a ray $\vecr$ intersecting $\sigma$ in its interior. Let $H^+(T)$ be the open half-space delimited by $\aff(T)$ and containing $\zero$.
Check if there is a colourful $d$-simplex having $\sigma$ as a facet and either containing $\zero$ or having a facet $\tau$ intersecting $\vecr$ before $\sigma$. If there is none, we obtain $d$ new vertices forming a $\widehat{d+1}$-transversal $T'=\{v'_1,\ldots,v'_d\}$ in $H^+(T)$, see Section~\ref{Sproofsub}. 
Take a point $x\notin\bigcup_i\S_i$ in $\aff(\vecr)$ such that $\zero\in\conv(v_1,\ldots,v_d,x)$, and choose any point  $v'_{d+1}\in\S_{d+1}$. 
We can use Proposition~\ref{prop2} and its constructive proof with $\S^*_i=\{v_i,v'_i\}$ for $i=1,\ldots,d$ and $\S^*_{d+1}=\{x,v'_{d+1}\}$ to obtain a new colourful simplex containing $\zero$ with at least one vertex in $T'$.  If $v'_{d+1}$ is a vertex of the new simplex, we do have a colourful simplex containing $\zero$. Otherwise, the facet of the simplex not containing $x$ is a colourful $(d-1)$-simplex $\tau$ intersecting  $\vecr$ before $\sigma$ since $\aff(T)$ forms the boundary of $H^+(T)$.\\

\noindent
Given a colourful $(d-1)$-simplex $\sigma$ intersecting $\vecr$, the proposed algorithm finds either a colourful simplex containing $\zero$, or a colourful $(d-1)$-simplex $\tau$ intersecting $\vecr$ before $\sigma$.  Since there is a finite number of colourful $(d-1)$-simplices, the algorithm eventually finds a colourful simplex containing $\zero$. 
While non-proven to be polynomial, pivot-based algorithms, such as the B\'ar\'any-Onn ones or our algorithm, are typically efficient in practice.\\\\

\noindent{\bf Acknowledgments$\:$}
This work was supported by grants from NSERC, MITACS, and Fondation Sciences Math\a'ematiques de Paris, and by the Canada Research Chairs program. We are grateful to Sylvain Sorin and Michel Pocchiola for providing the environment that nurtured this work from the beginning. 

\bibliographystyle{amsalpha}
\bibliography{refs}

\end{document}